\newif\ifextended
\extendedtrue

\documentclass[orivec]{llncs}
\usepackage{tkz-graph}
\usepackage[ruled,vlined,linesnumbered]{algorithm2e}
\usepackage{booktabs}
\usepackage[font=small]{caption}
\usepackage{xcolor}
\usepackage{soul}
\usepackage[square,sort,numbers]{natbib}
\usepackage{times}
\usepackage{balance}

\let\llncssubparagraph\subparagraph
\let\subparagraph\paragraph
\usepackage{titlesec}
\let\subparagraph\llncssubparagraph

\usepackage{paralist}

\usepackage{amsthm}
\usepackage{amsmath}
\usepackage{enumitem}
\usepackage{MnSymbol}
\usepackage{titling} 
\usepackage{authblk}
\usetikzlibrary{arrows}
\pgfdeclarelayer{background}
\pgfsetlayers{background,main}

\DeclareCaptionFont{footnotesize}{\footnotesize}

\newcommand{\upd}{\ensuremath{\mathit{upd}}}
\newcommand{\paths}{\ensuremath{\mathit{paths}}}
\newcommand{\maxpaths}{\ensuremath{\mathit{maxpaths}}}
\newcommand{\nodes}{\ensuremath{\mathit{nodes}}}
\newcommand{\out}{\ensuremath{\mathit{out}}}
\newcommand{\valid}{\ensuremath{\mathit{valid}}}
\newcommand{\true}{\ensuremath{\mathit{true}}}
\newcommand{\false}{\ensuremath{\mathit{false}}}
\newcommand{\up}{\ensuremath{\mathit{up}}}
\newcommand{\down}{\ensuremath{\mathit{down}}}
\newcommand{\wait}{\ensuremath{\mathit{wait}}}

\makeatletter
\newcommand\incircbin
{%
  \mathpalette\@incircbin
}
\newcommand\@incircbin[2]
{%
  \mathbin%
  {%
    \ooalign{\hidewidth$#1#2$\hidewidth\crcr$#1\bigcirc$}%
  }%
}

\newcommand{\thewait}[1]{{\small\incircbin{#1}}}
\makeatother

\makeatletter
\renewcommand\bibsection%
{
  \section*{\refname
    \@mkboth{\MakeUppercase{\refname}}{\MakeUppercase{\refname}}}
}
\makeatother

\newtheoremstyle{discstyle}
  {0.2em} %
  {0.2em} %
  {} %
  {} %
  {\bfseries} %
  {.} %
  {.5em} %
  {} %

\newtheorem*{rep@theorem}{\rep@title}
\newcommand{\newreptheorem}[2]{%
\newenvironment{rep#1}[1]{%
 \def\rep@title{#2 \ref{##1}}%
 \begin{rep@theorem}}%
 {\end{rep@theorem}}}
\makeatother

\theoremstyle{discstyle}
\newtheorem{lemmat}[lemma]{Lemma}
\newreptheorem{lemma}{Lemma}
\theoremstyle{discstyle}

\theoremstyle{discstyle}
\newreptheorem{theorem}{Theorem}

\newcommand{\cameraversion}[1]{\ifextended \else {#1} \fi}
\newcommand{\extendedversion}[1]{\ifextended {#1} \else \fi}

\newtheoremstyle{named}{0.2em}{0.2em}{}{}{\bfseries}{.}{.5em}{\thmnote{#3}}
\theoremstyle{named}
\newtheorem*{namedtheorem}{Theorem}

\newcommand{\makeproof}[8]{
\newcommand{#1}{
\ifthenelse{\equal{#7}{}}{
\begin{namedtheorem}[#4 \ref{#2}]
#5
\end{namedtheorem}
}{
\begin{namedtheorem}[#4 \ref{#2} \textnormal{(#7)}]
#5
\end{namedtheorem}
}
{\color{blue}
\begin{proof}
#6
\end{proof}}
}
\ifthenelse{\equal{#8}{}}{
\begin{#3}
\label{#2}
#5
\end{#3}
\extendedversion{\begin{proof} #6 \end{proof}}
}{
\begin{#3}[#7]
\label{#2}
#5
\end{#3}
\begin{proof} 
#6 
\end{proof}
}
}

\makeatletter
\g@addto@macro\normalsize{%
  \setlength\abovedisplayskip{0pt}
  \setlength\belowdisplayskip{0pt}
  \setlength\abovedisplayshortskip{0pt}
  \setlength\belowdisplayshortskip{0pt}
}
\makeatother

\ifextended
  \titlespacing{\section}{0pt}{0.6em}{0.4em}
  \titlespacing{\subsection}{0pt}{0.6em}{0.4em}
  \titlespacing{\paragraph}{0cm}{0.6em}{0.4em}
\else
  \titlespacing{\section}{0pt}{0.4em}{0.2em}
  \titlespacing{\subsection}{0pt}{0.4em}{0.2em}
  \titlespacing{\paragraph}{0cm}{0.4em}{0.3em}
\fi
\def\thespacing{0.5\baselineskip}
\makeatletter
\def\thm@space@setup{%
  \thm@preskip=0.2em
  \thm@postskip=\thm@preskip %
}
\makeatother
\newcommand{\optnewpage}{}

\title{Optimal Consistent Network Updates in Polynomial Time\extendedversion{\\{\Large (Extended Version)}}}
\date{}
\author[1]{Pavol {\v C}ern\'y}
\author[2]{Nate Foster}
\author[1]{Nilesh Jagnik} 
\author[1]{Jedidiah McClurg}
\affil[1]{University of Colorado Boulder} 
\affil[2]{Cornell University}
\begin{document}

\maketitle

\begin{abstract} 
Software-defined networking (SDN) allows operators to control
the behavior of a network by programatically managing
the forwarding rules installed on switches.
However, as is common in distributed systems, it can be difficult
to ensure that certain consistency properties are preserved during
periods of reconfiguration. 
The widely-accepted notion of {\em per-packet consistency}
requires every packet to be forwarded using the new
configuration or the old configuration, but not a mixture of the two.
If switches can be updated in some (partial) order which
guarantees that per-packet consistency is preserved, we call
this order a {\em consistent order update}.
In particular, switches that are incomparable in this order can
be updated in parallel. We call a consistent order update
{\em optimal} if it allows maximal parallelism.
This paper presents a polynomial-time algorithm for finding an optimal
consistent order update. This contrasts with other recent results in the literature,
which show that for other classes of properties (e.g., loop-freedom and waypoint
enforcement), the optimal update problem is {\sc np}-complete.
\end{abstract}

\section{Introduction}
\label{sec:intro}

Software-defined networking (SDN) replaces
conventional network management interfaces with higher-level APIs. 
While SDN has been used to build a wide variety of useful
applications, in practice, it can be difficult for operators to
{\em correctly} and {\em efficiently} reconfigure the network, i.e., update the
global set of forwarding rules installed on switches (known as a
\emph{configuration}). Even if the initial and final configurations
are free of errors, na{\"i}vely updating individual switches (referred
to in this paper as {\em switch-updates}) can lead to incorrect
transient behaviors 
such as forwarding loops, blackholes, bypassing a firewall, etc. In
certain cases, updating switches in parallel can lead to incorrect
transient behavior, but in other cases we can correctly parallelize switch updates.
Therefore, we need a partial order on switch-updates which
ensures that correctness properties hold before, during, and
after the update. 

\paragraph{Consistent order updates.}
This paper investigates the problem of computing a {\em
  consistent order update}. Given an initial and final network configuration,
a consistent order update is a partial order on switch-updates,
such that if the switches are updated according to this
order, an important consistency property called {\em per-packet
  consistency}~\cite{Reitblatt:2012:ANU:2342356.2342427} is guaranteed
throughout the update process. This property guarantees that each packet
traversing the network will follow a single global configuration:
either the initial one, or the final one, but not a mixture of the
two. In particular, this means that if the initial and the final
configurations are loop-free, blackhole-free, prevent bypassing a
firewall, etc., then so
do all intermediate configurations.

\paragraph{Optimal consistent order updates.}
In implementing a consistent order update, we would generally prefer
to use one that is optimal. A consistent order
update is {\em optimal} if it allows the most parallelism among all
consistent order updates. Formally, recall that a consistent order
update is a partial order on switch-updates---an optimal partial order
is one where the length of the longest chain in the order is the
smallest among all possible correct partial orders. Intuitively, this
means the update can be performed in the smallest number of
``rounds,'' where rounds are separated by waiting for in-flight
packets to exit the network and by waiting for all the switch updates
from the previous rounds to finish.

\paragraph{Single flow vs. multiple flows.}
A {\em flow} is a restriction of a network configuration to packets of
a single type, corresponding to values in packet headers. A packet
type might include the destination address, protocol number
(TCP vs. UDP), etc.
We show that if we consider flows to be {\em symbolic} (i.e. represented
by predicates over packet headers, potentially matching multiple
flows), then the problem is 
{\sc co-np}-hard. %
In this
paper, we focus on the problem of updating an {\em individual} flow---i.e.,
we are interested in the situation where the flows to be updated can be
enumerated. Furthermore, 
as we are looking for efficient consistent order updates, we focus on the
case where each switch can be updated at most once, from its initial to its final configuration.  

\paragraph{Main result.}
Our main result is that for updating a single flow, there is a
polynomial-time algorithm, with $O(n^2(n+m))$ complexity where $n$ is
the number of switches and $m$ the number of links. The
result is interesting both  
theoretically and practically. On the theoretical side, 
recent papers have presented complexity results for network
updates. However, for many other consistency properties
(loop-freedom, waypoint enforcement) and network models, the
optimal network update problem is {\sc np}-hard \cite{dudycz2016can,
  forsterconsistent, Ludwig:2015:SLN:2767386.2767412, LudwigRFS14,
  ludwig2016transiently, luo2016arrange}. The same is true for results
that study these problems with a model which is the same as ours
(single flows, update every switch at most
once). %
In contrast, we provide a {\em positive} result that there exists a
polynomial-time algorithm for optimal order updates for a single
flow, with respect to the per-packet consistency property. The
consistency properties studied in these papers (loop-freedom and
waypoint enforcement) are weaker than per-packet consistency, which
offers a trade-off: enforcing only (for instance) loop-freedom allows 
more updates to be found, but it is an (exponentially) harder problem.
In practice, network operators might wish to update only a small number
of flows, and here our polynomial-time algorithm would be
advantageous. A potential limitation is that if many flows
are considered separately, it could lead to large forwarding tables. 

\paragraph{Algorithm.}
Our algorithm models a network configuration as a directed graph with
unlabeled edges, and an update from an initial configuration to a
final configuration as a sequence of individual switch-updates---i.e.,
updating the outgoing edges at each switch. In order to determine
whether a switch $n$ can be updated while properly respecting the
per-packet consistency property, we define a set of conditions on the
paths {\em upstream} and {\em downstream} from $n$. We show that
these conditions can be checked in $O(n(n+m))$ time. In this way,
the algorithm produces a partial order on switches, representing the
consistent order update (if such an order does not exist, our algorithm
reports a failure). Additionally, we show that if the partial order is
constructed greedily (i.e., all nodes that can be updated
are immediately updated in parallel),
it results in an {\em optimal} consistent order update.
The challenging part of the proof is
to show that this algorithm is complete (i.e., always finds a
consistent order update if one exists) and optimal.

\section{Overview}
\label{sec:overview}

This section presents a number of simple examples to help develop intuition about the
consistent order updates problem and the challenges that any solution must address.

\paragraph{Consistent order updates.}
\begin{figure}[t]
\centering
\begin{minipage}{.4\textwidth}
  \centering
  \tikzset{ 
	VertexStyle/.append style = { minimum size = 12pt, inner sep = 0pt, font =\tiny }
	}
  \begin{tikzpicture}
  \Vertex[L=$H_1$]{H1}
  \NO[unit=1](H1){A}
  \NO[unit=1](A){C}
  \EA[unit=2,L=$H_2$](H1){H2}
  \NO[unit=1](H2){B}
  \NO[unit=1](B){D}

  \draw[-triangle 45](H1.north west) -- (A.south west);
  \draw[dashed, -triangle 45](H1.north east) -- (A.south east);

  \draw[-triangle 45](A) -- (C);
  \draw[-triangle 45](B.south west) -- (H2.north west);
  \draw[dashed, -triangle 45](B.south east) -- (H2.north east);

  \draw[-triangle 45](C) -- (B);
  \draw[dashed, -triangle 45](A) -- (D);
  \draw[dashed, -triangle 45](D) -- (B);
  \end{tikzpicture}
  \caption{Trivial update.}%
  \label{fig:easy}
\end{minipage}%
\begin{minipage}{.6\textwidth}
  \centering
  \tikzset{ 
	VertexStyle/.append style = { minimum size = 12pt, inner sep = 0pt, font =\tiny }
	}
  \begin{tikzpicture}
  \Vertex[L=$H_1$]{H1}
  \NOEA(H1){A}
  \SOEA(H1){B}
  \NOEA(B){C}
  \NOEA(C){D}
  \SOEA(C){E}
  \NOEA[L=$H_2$](E){H2}

  \draw[-triangle 45](H1.south east) -- (B.north west);
  \draw[dashed,-triangle 45](H1.north east) -- (A.south west);
  \draw[dashed,-triangle 45](A.south east) -- (C.north west);
  \draw[-triangle 45](B.north east) -- (C.south west);
  \draw[dashed,-triangle 45](C.north east) -- (D.south west);
  \draw[-triangle 45](C.south east) -- (E.north west);
  \draw[dashed,-triangle 45](D.south east) -- (H2.north west);
  \draw[-triangle 45](E.north east) -- (H2.south west);
  \end{tikzpicture}
  \caption{Double diamond: no consistent update order exists.}
  \label{fig:doubled}
\end{minipage}
\end{figure}
Consider Figure~\ref{fig:easy}. In the initial configuration $C_i$
(denoted by {\em solid} edges), the forwarding-table rules (outgoing edges) on
each switch are set up such that host $H_1$ is sending
packets to $H_2$ along the path $H_1 {\rightarrow} A {\rightarrow} C {\rightarrow} B {\rightarrow} H_2$.
Let us assume that switch $C$ is scheduled for maintenance, meaning we must first
transition to configuration $C_f$ (denoted by the {\em dashed} edges).
Note that the two configurations differ only for nodes $A$ and
$D$. If the node $A$ is updated before node $D$, packets from $H_1$ will be
dropped at $D$. On the other hand, updating $D$ before $A$ leads to a
consistent order update. 
Note that since we model networks as graphs, we will
use the terms {\em switch} and {\em node} interchangeably based on the context,
and similarly for the terms {\em edge} and {\em forwarding rule}.
{\em Path} will be used to describe a sequence of adjacent edges.

In Figure~\ref{fig:doubled}, regardless of the order in which we update
nodes, there will always be inconsistency. Note that here the nodes
$A$ and $D$ can be updated first, but a problem arises due to nodes $H_1$
and $C$. Specifically, if $C$ is updated before $H_1$, then the network is in a
configuration containing a path $H_1{\rightarrow}B{\rightarrow}C{\dashrightarrow}D{\dashrightarrow}H_2$,
which is not in
either $C_i$ or $C_f$. In other words, $H_1$ cannot be updated
unless the (downstream) path from $C$ to $H_2$ is first updated.
On the other hand, $C$ cannot be updated unless the (upstream)
path from $H_1$ to $C$ is first
updated. We refer to this case as a {\em double diamond}. If we consider
the notion of dependency graphs \cite{MahajanW13}, where there is an edge
from a node $x$ to node $y$ if the update of $y$ can only be executed after
the update of $x$, then our double diamond example corresponds to a
cyclic dependency graph between $H_1$ and $C$. 

Unfortunately, the presence of a double diamond (cyclic
dependency) does not necessarily indicate that there cannot be a
solution. Consider Figure~\ref{fig:rem_dd}, where there
is a double diamond between $D$ and $J$. Updating $B$ removes the old
traffic to $D$, and then after updating $B$, the nodes
$D,E,G,F,H,I,J$ have no incoming traffic. At this point, these
nodes can be updated without violating per-packet
consistency. Thus, the circular dependency has been eliminated,
allowing a valid update order such as
$[A,H_1,K,L,B,D,E,F,G,H,I,J,\allowbreak C,M]$. This shows that an approach (such
as \cite{yuan2014generating,Jin:2014:DSN:2619239.2626307}) based on a static dependency
graph might miss some cases where a consistent order update exists---a
limitation that our algorithm does not exhibit.

\begin{figure}[h]
\centering
\begin{minipage}{.6\textwidth}
  \centering
  \tikzset{ 
	VertexStyle/.append style = { minimum size = 12pt, inner sep = 0pt, font =\tiny }
	}
\begin{tikzpicture}
\Vertex[L=$H_1$]{H1}
\EA(H1){A}
\NOEA[unit=0.71](A){B}
\SOEA[unit=0.71](A){C}
\SOEA[unit=0.71](B){D}
\NOEA[unit=0.71](D){E}
\SOEA[unit=0.71](D){F}
\SOEA[unit=0.71](E){G}
\NOEA[unit=0.71](G){H}
\SOEA[unit=0.71](G){I}
\SOEA[unit=0.71](H){J}
\EA(J){K}
\EA[L=$H_2$](K){H2}
\NO[unit=0.8](G){L}
\SO[unit=0.8](G){M}

\draw[dashed,-triangle 45](H1.north east) -- (A.north west);
\draw[-triangle 45](H1.south east) -- (A.south west);

\draw[dashed,-triangle 45](A.north) -- (B.west);
\draw[-triangle 45](A.north east) -- (B.south west);

\draw[dashed,-triangle 45](A.south east) -- (C.north west);
\draw[-triangle 45](A.south) -- (C.west);

\draw[-triangle 45](B.south east) -- (D.north west);

\draw[-triangle 45](D.north east) -- (E.south west);
\draw[dashed,-triangle 45](D.south east) -- (F.north west);

\draw[dashed,-triangle 45](C.north east) -- (D.south west);
\draw[-triangle 45](E.south east) -- (G.north west);
\draw[dashed,-triangle 45](G.south east) -- (I.north west);
\draw[-triangle 45](G.north east) -- (H.south west);
\draw[-triangle 45](H.south east) -- (J.north west);

\draw[dashed, -triangle 45](F.north east) -- (G.south west);
\draw[dashed, -triangle 45](I.north east) -- (J.south west);

\draw[dashed,-triangle 45](J.north east) -- (K.north west);
\draw[-triangle 45](J.south east) -- (K.south west);

\draw[dashed,-triangle 45](K.north east) -- (H2.north west);
\draw[-triangle 45](K.south east) -- (H2.south west);

\draw[dashed,-triangle 45](B) to[out=45,in=135] (L);
\draw[dashed,-triangle 45](L) to[out=45,in=90] (K);

\draw[-triangle 45](C) to[out=-45,in=-135] (M);
\draw[-triangle 45](M) to[out=-45,in=-90] (K);

\end{tikzpicture}\vspace{-8pt}
\caption{Removable double diamond.}
\label{fig:rem_dd}
\end{minipage}
\begin{minipage}{.35\textwidth}
  \centering
  \tikzset{ 
	VertexStyle/.append style = { minimum size = 12pt, inner sep = 0pt, font =\tiny }
	}
  \begin{tikzpicture}
  \Vertex[L=$H_1$]{H1}
  \EA[unit=1.5](H1){B}
  \NOEA[unit=1.25](H1){A}
  \SOEA[unit=1.25](H1){C}
  \draw[dashed,-triangle 45](H1.north) -- (A.west);
  \draw[-triangle 45](H1.north east) -- (A.south west);
  
  \draw[dashed,-triangle 45](H1.south east) -- (C.north west);
  \draw[-triangle 45](H1.south) -- (C.west);
  
  \draw[dashed,-triangle 45](H1.north east) -- (B.north west);
  \draw[-triangle 45](H1.south east) -- (B.south west);
  
  \draw[dashed,-triangle 45](C) -- (B);
  \draw[-triangle 45](A) -- (B);
  
  \draw[dashed,-triangle 45] (B) -- (5:2.5cm);
  \draw[-triangle 45] (B) -- (-5:2.5cm);

  \draw[dashed,-triangle 45] (A) -- (35:2.5cm);

  \draw[-triangle 45] (C) -- (-35:2.5cm);

  \end{tikzpicture}
  \caption{Wait example.}
  \label{fig:waitdual}
\end{minipage}%
\end{figure}

\paragraph{Waits.}
As mentioned, it may be impossible to
parallelize certain updates---we may need to make sure that
some node $x$ is updated before another node $y$. We may need to {\em wait} 
during the sequence of switch-updates to ensure that such updates are
executed one after the other. 
This requirement can arise because when updating a node,
we may need to ensure that (1) all of the previous switch-updates have been completed, and
(2) all of the packets that were in the network since before the previous update have exited
the network.
The former type we call a {\em switch-wait}, and the latter a {\em packet-wait}. 

In Figure~\ref{fig:rem_dd}, we see that $L$ must be updated before
updating $B$. To ensure that edges outgoing from $L$ are ready,
we must wait after sending the update command to $L$, in
order to ensure that its forwarding rules have been fully installed.
In other words, we say that there is a {\em switch-wait} required between updates of $L$ and $B$.
After updating $B$, the switch $D$ becomes disconnected, but there may still be some
packets in transit on the $B {\rightarrow} D$ path. Before updating $D$,
we must ensure that packets along these old removed paths have been
flushed from the network. For this reason, we need a {\em packet-wait}
between updates of nodes $D$ and $B$.

If we are interested only in
finding a correct sequence of updates, we can wait (for an amount of
time larger than the maximum switch-wait and packet-wait duration)
after every node update.
However, waits may not be necessary after every update if we
update switches from separate parts of the network. For the Figure~\ref{fig:rem_dd} example,
the correct sequence with a {\em minimal} number of waits is
$[A,H_1,K,L,\thewait{s},B,\thewait{p},D,E, F,\allowbreak G, H,I,J,\thewait{s},C,M]$, where $\thewait{p}$ denotes
a packet-wait and $\thewait{s}$ denotes a switch-wait. In this example, nodes
$A$, $H_1$, $K$, $L$ can be updated in parallel. Similarly, nodes
$D$, $E$, $F$, $G$, $H$, $I$ can be updated in parallel, etc. There are
three waits, meaning this consistent order update
requires {\em four} switch-update {\em rounds}.

The example in Figure~\ref{fig:waitdual} highlights the relationship between
switch-waits and packet-waits.
Observing that the configurations are roughly symmetrical,
let us examine the relationship between nodes $A$, $B$, $C$.
The correct order of updates between these nodes is
$H_1,A,\thewait{p},B,\thewait{s},C$.
There must be a {\em switch-wait} between the updates of $B$ and $C$,
due to the presence of a $C_f$ path $C{\dashrightarrow}B$.
There must be a {\em packet-wait} between updates of switches $A$ and $B$,
due to the presence of a $C_i$ path $A{\rightarrow}B$. 

As is common elsewhere (e.g. \cite{LudwigRFS14}), in this paper,
we do not distinguish between 
packet-waits and switch-waits, and only use the term {\em wait}---our goal is
to maximize the parallelism of switch-updates,
i.e. minimize the number of switch-update rounds. 

\section{Network Model}
\label{sec:model}

\paragraph{Network and Configurations.}   A topology of a network is a
graph $G = (N,E)$, where $N$ is a set of nodes, and $E$ is a set of
directed edges. A configuration $C \in \mathcal{P}(E)$ is a subset of
edges in $E$. A {\em proper} configuration is 
such that (a) it has one source $H_1$ and (b) it is acyclic. Here, a source is a
designated node with no incoming edges, representing the point where
packets enter the network.  Note that cycles in a configuration are
undesirable, as this would mean that traffic might loop forever in the
network.    
We first consider the
case with one source, and in Section~\ref{sec:discussion}, we describe
a simple reduction for the case of multiple sources. 
Our goal is to transition from an initial configuration $C_i$
to a final configuration $C_f$ by updating individual nodes. Consider
$C_i$ and $C_f$ to be fixed throughout the paper, and 
assume both are proper.

\paragraph{Updates.}   
Let $u$ be a node, and let $C$ be a configuration. We define a function
$\out(C,u)$ which returns the set of edges from $C$ whose source is $u$. 
The function $\upd_1(C,u)$ returns the configuration $C'$ such that 
$C' = (C \setminus \out(C_i,u)) \cup \out(C_f,u)$. That is, $C'$ has the node $u$
updated to the final configuration. 
Let $R$ be the set of all sequences that can be
formed using nodes in $N$ without repetition. We
extend $\upd_1$ to sequences of nodes by defining the function $\upd$
that, given a configuration $C$ and a sequence of 
nodes $S$, returns a  configuration $C'=\upd(C,S)$. The function $\upd$
is defined by $\upd(C,\varepsilon)=C$ (where $\varepsilon$ is the empty
sequence), and
$\upd(C,uS)=\upd(\upd_1(C,u),S)$. We consider sequences of nodes
without repetition, because our goal is to find update sequences that
update every node at most once.

\paragraph{Paths.} Given a configuration $C$, a $C$-path is a directed
path (finite or infinite) whose edges are in $C$. For a path $p$, we
write $p \in C$ if $p$ is a $C$-path. 
A $C_i$-only path is one which is in
$C_i$ and not in  
$C_f$. Similarly, a $C_f$-only path is in $C_f$ but
not $C_i$. The function $\nodes$ takes a path $q$ as an argument and returns
a set $Q$ of all nodes on a path. 
Let $s$ and $t$ be two nodes, and let $C$ be a
configuration. The function $\paths(s,t,C)$ returns the 
set of all paths between $s$ and $t$ in configuration
$C$. A path $p$ in a configuration $C$ is {\em maximal} if it is
either (a) finite, and its last node has no outgoing edges in $C$, or
(b) infinite. 
The function $\maxpaths(s,C)$ returns the set of all maximal paths
starting at node $s$ in configuration $C$.

\paragraph{Path and Configuration Consistency.}

We say that a path $p$ is \emph{consistent} if $p \in
\maxpaths(H_1,C_i) \lor p \in \maxpaths(H_1,C_f)$, and a configuration
$C$ is {\em consistent} if and only if $\forall p \in
\maxpaths(H_1,C)$, we have that $p$ is consistent. Intuitively, all
maximal paths starting at $H_1$ are maximal paths in either the old
configuration or the new configuration---%
this corresponds to per-packet consistency~\cite{Reitblatt:2012:ANU:2342356.2342427}. If initial
configuration $C_i$ and final configuration $C_f$ are proper, then
so is every consistent configuration. 

\paragraph{Waits.}  Let $U=u_1u_2\cdots u_k$ be a sequence of node
updates. Let $C_j=\upd(C_i,U_j)$ be the configuration
reached after updating a sequence $U=u_1u_2\cdots u_j$ for 
$1 \leq j \leq k$, and let $C_0=C_i$.
For $l,u$ such that $0 \leq l \leq u \leq k$, let $C^u_l$ be the
configuration obtained as a union of configurations 
$C_l \cup \cdots \cup C_u$.
We say that a {\em wait is needed} between $u_j$ and $u_k$ in U if and only if the
configuration $C^k_{j-1}$ is not consistent. 
To illustrate, let us return to the example in
Figure~\ref{fig:waitdual} (note that we no longer distinguish
between packet-waits and switch-waits). As mentioned,
after updating $H_1$ and $A$, we need a wait before updating $B$. Let
the configuration $C_v$ be the union of all the intermediate
configurations until after the update to $B$. Then $C_v$
has the path $H_1 {\rightarrow} A {\rightarrow} B {\rightarrow}$,
where we take the solid edge from $A$ to $B$ and a dashed outgoing
edge from $B$, meaning a wait is needed. In this case, using the union of the
configurations captures the reason for the wait.

\paragraph{Consistent update sequence.}
For any set of nodes $S$,
let $\pi(S)$ be the set of sequences that can be formed by nodes in
$S$, without repetition. 
Let $Z = S_1S_2 \cdots S_k$ be a sequence such that each $S_i$ is a subset of
$N$. 
Let $\pi(Z)$ be the set of sequences defined by 
$\{ r_1r_2 \cdots r_k \mid r_1 \in \pi(S_1) \land  r_2 \in \pi(S_2)
\land \cdots \land r_k \in \pi(S_k) \}$. 
 
The sequence $Z = S_1 S_2 \cdots S_k$ is a {\em consistent update
  sequence} if and only if
\begin{compactenum}
  \item The sets $S_1,S_2,\cdots,S_k$ form a partition of the set of nodes
    $N$. Note that this ensures that $\forall U \in \pi(Z)$, we have
    $\upd(C_i,U)=C_f$, i.e., after updating $u$, we are in 
    $C_f$. 
  \item $\forall U \in \pi(Z)$, for every prefix $U'$ of $U$,
    $C{=}\upd(C_i,U')$ is a consistent configuration. 
  \item $\forall U \in \pi(Z)$, let $U'=u_1u_2\cdots u_j$ and
    $U''=u_1u_2\cdots u_k$ be prefixes of $u$, s.t. $k>j$, then if
    a wait is needed between $u_j, u_k$ in $U$, then $u_j, u_k$ are in different sets $S$ and $S'$. 
\end{compactenum}

\paragraph{Consistent Order Update Problem.}   
Given an initial configuration $C_i$ and the final configuration
$C_f$, the {\em consistent order update problem} is to find a consistent update
sequence if there exists one. 

\paragraph{Optimal Consistent Order Update Problem.}
Given $C_i$ and 
$C_f$, if a consistent update sequence exists, the {\em optimal
  consistent update problem} is to find a consistent update sequence
of minimal length. 

\section{OrderUpdate Algorithm}
\label{sec:algo}
\bgroup
\def\arraystretch{1.5}
\begin{figure}[t]
\footnotesize
\setlength\tabcolsep{8pt}
\begin{center}
  \begin{tabular}{|p{0.01\textwidth}|p{0.44\textwidth}|p{0.44\textwidth}|}
  	\hline
    & Upstream $\newline$(Condition for $\paths(H_1,s,C_c)$)  & Downstream $\newline$(Condition for $\maxpaths(s,C_c)$)\\\hline
    A   &   $Y_a(s) = \not\exists p \in \paths(H_1,s,C_c)$   &   $Z_a^{\dagger}(s) = (\out(s,C_f)=\emptyset)\lor \newline\hspace*{1.15cm} \forall p \in \maxpaths(s,\upd(C_c,s)):\hspace*{1.25cm} p \in \maxpaths(s,C_f)$\\\hline
    B & $Y_b(s) = \lnot Y_a(s) \land \forall p \in \paths(H_1,s,C_c) : \hspace*{1.0cm} p \in \paths(H_1,s,C_i) \newline\hspace*{0.75cm}\land p \in \paths(H_1, s, C_f)$   &   $Z_b(s) = \forall p \in \maxpaths(s,upd(C_c,s)): \hspace*{1.25cm} p \in \maxpaths(s,C_i) \newline\hspace*{1.0cm} \lor p \in \maxpaths(s,C_f)$\\\hline
    C & $Y_c(s) = \lnot Y_a(s) \land \lnot Y_b(s) \newline\hspace*{0.5cm} \land \forall p \in paths(H_1,s,C_c) : \newline\hspace*{1.025cm} p \in \paths(H_1,s,C_f)$   &   $Z_c(s) = \forall p \in \maxpaths(s,upd(C_c,s)): \hspace*{1.25cm} p \in \maxpaths(s,C_f)$ \\\hline
    D & $Y_d(s) = \lnot Y_a(s) \land \lnot Y_b(s) \newline\hspace*{0.5cm} \land \forall p \in paths(H_1,s,C_c) : \newline\hspace*{1.025cm} p \in \paths(H_1,s,C_i)$   &   $Z_d(s) = \forall p \in \maxpaths(s,upd(C_c,s)): \hspace*{1.25cm} 	p \in \maxpaths(s,C_i)$\\\hline
    E & $Y_e(s) = \hspace*{0.2cm} \lnot Y_a(s) \land \lnot Y_b(s) \newline\hspace*{1cm} \land \lnot Y_c(s) \land \lnot Y_d(s)$
    $\newline\hspace*{0.75cm}=(\exists p_{f} \in \paths(H_1,s,C_c): \newline\hspace*{1.4cm} p_{f} \in \paths(H_1,s,C_f) \newline\hspace*{1.05cm} \land p_{f} \not\in \paths(H_1,s,C_i)) \newline\hspace*{0.75cm} \land (\exists p_{i} \in \paths(H_1,s,C_c):\newline\hspace*{1.4cm} p_{i} \in \paths(H_1,s,C_i) \newline\hspace*{1.05cm} \land p_{i} \not\in \paths(H_1,s,C_f))$
    & $Z_e(s) = \forall p \in \maxpaths(s,upd(C_c,s)): \hspace*{1.25cm} p \in \maxpaths(s,C_i) \newline\hspace*{1.0cm} \land p \in \maxpaths(s,C_f) $\\
    \hline
  \end{tabular}
\end{center}
\captionsetup{font=footnotesize}
\caption{Necessary conditions for updating a node $s$ in current configuration $C_c$}%
\label{fig:condtable}
\end{figure}
\egroup

This section presents an algorithm (Algorithm~\ref{alg:orderupdate})
that solves the consistent order update problem. It works by
repeatedly finding and updating a node that can be updated without
violating consistency. For clarity, we focus first on
correctness. Section~\ref{sec:waits} presents an improved version that
finds an optimal update.

\paragraph{Correct Sequence.}   
A \emph{correct} sequence of node updates $T=t_1t_2\cdots t_{|N|}$
refers to a consistent update sequence of singleton sets
$Z=S_1S_2\cdots S_{|N|}$ s.t. $\forall j \in [1,|N|] : S_j =
\{t_j\}$. Algorithm~\ref{alg:orderupdate} uses a subroutine at
Line~\ref{alg:line:pick} (in this section, the subroutine is
Algorithm~\ref{alg:pickdef}; in Section~\ref{sec:waits} we will
replace it with Algorithm~\ref{alg:pickandwait} to achieve
optimality) to find a correct update sequence. It takes $C_i, C_f$
as input and returns two sequences of nodes, $R, R_w$. 
Sequence $R$ is the solution to the consistent order update problem
(a sequence of singleton sets). Sequence $R_w$
contains information about the placement of waits, which will be the
same as $R$ in this section, since we initially wait after every node
update.

\subsection{Necessary Conditions for Updating a Node}

To determine which node updates lead to consistent configurations, we
assume that the network is in a consistent configuration $C_c$, and
identify a set of necessary conditions which must hold in order for
the update to preserve consistency.
We classify nodes into five categories based on the types of paths
that are incoming to them from $H_1$. The classification is given in
the left-hand side of Figure~\ref{fig:condtable}.

\paragraph{Upstream Paths and Candidate Nodes.}   
Paths from source $H_1$ to a node $s$ are called \emph{upstream} paths
to $s$ (in some configuration). The condition on these paths is called
the upstream condition. If a node satisfies the upstream condition for
one of the five categories/types, it is known as a \emph{candidate} of
that type. 

\paragraph{Downstream Paths and Valid Nodes.}   
Downstream paths from a node $s$ are maximal paths starting at $s$ (in
some configuration). For each of the upstream conditions, there is a downstream condition which must
be satisfied, in order to ensure that all maximal paths starting from
$H_1$ in $\upd(C_c,s)$ through $s$ are consistent. If a candidate node
satisfies the corresponding downstream condition, it is called
\emph{valid}. A node which is not valid is called
\emph{invalid}. Note that upstream paths to $s$ are the same in $C_c$ and
$\upd(C_c,s)$.

\makeproof{\lemcorrectness}{lemma:correctness}{lemmat}{Lemma}
{In a consistent configuration $C_c$, if a valid node $s$ is updated,
then $\upd(C_c,s)$ is consistent. 
}
{
\begin{figure}
\centering
\begin{minipage}{.5\textwidth}
  \centering
  \tikzset{ 
	VertexStyle/.append style = { minimum size = 12pt, inner sep = 0pt, font =\tiny, color=black }
	}
  \begin{tikzpicture}
  \SetGraphColor{white}{blue}{magenta}
  \Vertex[L=$H_1$]{H1}
  \EA[L=$s$](H1){B}
  \EA[L=$ $](B){D}
  \NOEA[L=$ $](H1){A}
  \SOEA[L=$ $](H1){C}
  \EA[L=$ $](A){E}
  \EA[L=$ $](C){F}

  \draw[dashed,-triangle 45](H1.20) -- (B.160);
  \draw[-triangle 45](H1.-20) -- (B.-160);
  
  \draw[dashed,-triangle 45](B.20) -- (D.160);
  \draw[-triangle 45](B.-20) -- (D.-160);

  \draw[dashed,-triangle 45](B) -- (A);
  \draw[dashed,-triangle 45](A) -- (E);
  \draw[-triangle 45](B) -- (C);
  \draw[-triangle 45](C) -- (F);
  
  \begin{pgfonlayer}{background}
  \draw[line width = 5pt, gray!30](H1) -- (B);
  \draw[line width = 5pt, gray!30](B) -- (D);
  \draw[line width = 5pt, gray!30](B) -- (C);
  \draw[line width = 5pt, gray!30](C) -- (F);
  \draw[line width = 5pt, gray!30](A) -- (E);
  \end{pgfonlayer}{background}
  
  \end{tikzpicture}
  \caption{Type B Valid Node.}
  \label{fig:typeB}
\end{minipage}%
\begin{minipage}{.5\textwidth}
  \centering
  \tikzset{ 
	VertexStyle/.append style = { minimum size = 12pt, inner sep = 0pt, font =\tiny, color=black }
	}
  \begin{tikzpicture}
  \SetGraphColor{white}{blue}{magenta}
  
  \Vertex[L=$H_1$]{H1}
  \NOEA[L=$ $](H1){D}
  \SOEA[L=$ $](H1){E}
  \NOEA[L=$s$](E){S}
  \EA[L=$ $](S){B}
  \NOEA[L=$ $](S){A}
  \SOEA[L=$ $](S){C}

  \draw[dashed,-triangle 45](H1) -- (D);
  \draw[dashed,-triangle 45](D) -- (S);
  \draw[-triangle 45](H1) -- (E);
  \draw[-triangle 45](E) -- (S);
  \draw[dashed,-triangle 45](S.north) -- (A.west);
  \draw[-triangle 45](S.north east) -- (A.south west);
  
  \draw[dashed,-triangle 45](S.south east) -- (C.north west);
  \draw[-triangle 45](S.south) -- (C.west);
  
  \draw[dashed,-triangle 45](S.20) -- (B.160);
  \draw[-triangle 45](S.-20) -- (B.-160);
  
  \begin{pgfonlayer}{background}
  \draw[line width = 5pt, gray!30](H1) -- (E);
  \draw[line width = 5pt, gray!30](E) -- (S);
  \draw[line width = 5pt, gray!30](H1) -- (D);
  \draw[line width = 5pt, gray!30](D) -- (S);
  \draw[line width = 5pt, gray!30](S) -- (B);
  \draw[line width = 5pt, gray!30](S.60) -- (A.-150);
  \draw[line width = 5pt, gray!30](S.-60) -- (C.150);
  \end{pgfonlayer}{background}

  \end{tikzpicture}
  \caption{Type E Valid Node.}
  \label{fig:typeE}
\end{minipage}%
\end{figure}
Given a consistent configuration $C_c$, $\forall p \in
\maxpaths(H_1,\upd(C_c,s)): s \not\in \nodes(p) \rightarrow p \in
\maxpaths(H_1,C_c)$. Maximal paths that are not touched by $s$ are
retained from $C_c$ in $\upd(C_c,s)$. From consistency of $C_c$, these
paths are consistent. For checking the consistency of $\upd(C_c,s)$,
it is enough to ensure that $\forall p \in \maxpaths(H_1,upd(C_c,s)) :
s \in \nodes(p) \rightarrow p$ is consistent. We use this in the rest
of the proof. 
Our necessary conditions
for updating a node ensure that all maximal paths, starting from
$H_1$, in $upd(C_c,s)$ through $s$ are consistent. 
Figure~\ref{fig:condtable} identifies nodes as Types A-E based on
upstream conditions. The upstream conditions are
exhaustive and mutually exclusive, meaning each node is a candidate of
exactly one of the types. For each type, we show that if the
node is valid, then updating it preserves consistency. 
\begin{compactitem}
  \item Type A: no upstream paths incoming to node $s$ in
    $C_c$. Type A candidate nodes are also called a \emph{disconnected}
    nodes. Updating $s$ does not add downstream maximal paths starting
    from $H_1$ to $C_c$.  So, $\maxpaths(H_1,C_c) =
    \maxpaths(H_1,\upd(C_c,s))$, meaning updating $s$ preserves
    consistency. However, to simplify the presentation,
    Algorithm~\ref{alg:orderupdate} imposes a downstream condition. We
    will show that if a correct sequence exists, then there also
    exists some correct sequence that updates nodes with this optional
    downstream condition ($Z_a$ in Figure~\ref{fig:condtable}). 
  \item Type B: paths to $s$ from $H_1$ in $C_c$, are in both
    $\paths(H_1,s,C_i)$ and $\paths(H_1,s,C_f)$. Downstream paths in
    $\upd(C_c,s)$ from $s$ must be in either $\maxpaths(s,C_i)$ or
    $\maxpaths(s,C_f)$. This $s$ is a Type B valid node in
    Figure~\ref{fig:typeB}, where highlighted edges are in $C_c$. 
  \item Type C: all paths to $s$ from $H_1$ in $C_c$, are
    $\paths(H_1,s,C_f)$. To ensure consistency of $\upd(C_c,n)$,
    downstream maximal paths from $s$ in $\upd(C_c,s)$ must lie in
    $\maxpaths(s,C_f)$. 
  \item Type D: all paths to $s$ from $H_1$ in $C_c$, are
    $\paths(H_1,s,C_i)$. To ensure consistency of $\upd(C_c,n)$,
    downstream maximal paths from $s$ in $\upd(C_c,s)$ must lie in
    $\maxpaths(s,C_i)$. 
  \item Type E: some non-empty set of upstream paths to $s$ in $C_c$,
    are in $\paths(H_1,s,C_f)\setminus\paths(H_1,s,C_i)$, and some
    non-empty set of upstream paths to $s$ are in
    $\paths(H_1,s,C_i)\setminus\paths(H_1,s,C_f)$. This $s$ is a Type E
    valid node in Figure~\ref{fig:typeE}, where highlighted edges are
    in $C_c$. Downstream paths from $s$ in $\upd(C_c,s)$ must be in
    both $\maxpaths(s,C_i)$ and $\maxpaths(s,C_f)$.  
\end{compactitem}\vspace{-\baselineskip} 
}{}{blah}

Using Lemma~\ref{lemma:correctness}, each node updated by
OrderUpdate leads to a valid intermediate configuration. So, we
change from $C_i$ to $C_f$ without
going through an inconsistent state, and since we wait between all 
updates, we obtain a consistent sequence. 
 
\makeproof{\thmcorrectness}{th:correctness}{theoremt}{Theorem}
{%
Any sequence $R$ of nodes produced by Algorithm~\ref{alg:orderupdate}
(using subroutine Algorithm~\ref{alg:pickdef}) is correct. 
}
{
Every node updated by OrderUpdate preserves consistency in the
network. Let a sequence $S=s_1\cdots s_{|N|}$ be generated by
OrderUpdate. Then, using Lemma~\ref{lemma:correctness}, $\forall r \in
[1,|N|] : \upd(C_i,s_1\cdots s_{r-1})$ is consistent. Finally, since all
nodes are updated in $S$, $\upd(C_i,S) = C_f$. So, if a sequence of
updates is generated by Algorithm~\ref{alg:orderupdate} using
subroutine SequentialPickAndWait, it is a correct sequence. 
}{}{}

\cameraversion{{\noindent
The proof of this and other theorems/lemmas are in the
extended version \cite{nilesh}.}}

\subsection{Careful Sequences}
\label{sec:algo:careful}

Previously, we said that Type A candidates (disconnected nodes) do not
require a downstream condition to be updated. However,
Algorithm~\ref{alg:orderupdate} imposes a downstream condition on
disconnected nodes for them to be valid and updated. We refer to
sequences that respect this downstream condition (i.e., update only
valid nodes) as {\em careful}
sequences. Let $s$ be a node and $C$ be a configuration, and define
$\valid_1(C,s)$ to be $\true$ iff $s$ in valid in
configuration $C$. We extend $\valid_1$ to a sequence of nodes by defining
$\valid$ as $\valid(\varepsilon,C) = \true$ (where
$\varepsilon$ is the empty sequence) and $\valid(C,uS) =
\valid(\upd(C,u),S) \land \valid_1(C,u)$.  

\paragraph{Careful Sequence}
A \emph{careful} sequence $T=t_1t_2\cdots t_{|N|}$ is a correct sequence
of nodes s.t. $\forall l \in [1,|N|] :
\valid(\upd(C_i,t_1t_2\cdots t_{l-1}),t_l)$. 

\extendedversion{
Type A candidates do not have to be valid to be updated, but we
enforce the downstream condition for them to be valid. The downstream
condition for a Type A valid node $s$ in Figure~\ref{fig:condtable}
has two clauses: 
\begin{compactitem}
\item The first clause (final-connectivity condition) is true when $s$ is connected in $C_i$, but disconnected in $C_f$. If there are no outgoing $C_f$ edges from $s$ after its update, then it is a node which will be disconnected in $C_f$. After $s$ becomes disconnected,
it remains disconnected, as it has no incoming/outgoing $C_f$ edges, and can be updated.
\item The second clause states that all maximal paths downstream, after update, are in $\maxpaths(s,C_f)$.
This simplifies the proof of claims about correct sequences.
\end{compactitem}

\noindent
We will now prove that if there exists a correct sequence of updates,
then there is also a careful sequence of updates. Before proving this,
we first observe the following properties of correct sequences:

\begin{property}
\label{updateorder}
If we have two sequences $A$ and a permutation $A'$ of $A$ s.t
$\valid(C,A) \land valid(C,A')$, then $\upd(C,A) = \upd(C,A')$. 
\end{property}
\begin{proof}
This is because $A$ and $A'$ both update the same nodes in the graph. Additionally,
the final configuration after both updates has the same edges
regardless of the update order between $A$ and $A'$. 
\end{proof}

\begin{lemmat}
\label{lemma:equiv_main}
Let $T=UnV$ be a correct sequence where $n$ is an invalid Type A
candidate, then $\exists T'=Un'V'$, a correct sequence in which $n'$
is a valid node, and $V'$ is a sequence s.t. $n'V'$ is a permutation
of $nV$. 
\end{lemmat}
\begin{proof}
If $n$ is an invalid disconnected node, it was not disconnected in
$C_f$ (final-connectivity condition). Let $v_p$ be the first node in
sequence $V=v_1v_2\cdots v_k$ s.t. there is a path from $H_1$ to $n$ in
$\upd(C_i,Unv_1v_2\cdots v_p)$. Let us consider a sequence $V'' =
Uv_1v_2\cdots v_{p-1}nv_p\cdots v_k$.  
Let us define $\forall r \in [1,p) : C_r=\upd(C_i,Unv_1\cdots v_r)$ and
  $C_r'=\upd(C_i,Uv_1\cdots v_r))$. $\forall r \in [1,p) :
    \maxpaths(H_1,C_r)=\maxpaths(H_1,C_r')$ because there is no path
    from $H_1$ to $n$ in all configurations $C_r$ and $C_r'$. So in
    $V''$, updates of nodes $v_1,v_2,\cdots,v_{p-1}$ lead to consistent
    configurations. In $V''$, $n$ was disconnected before $v_p$ was
    updated, so updating $n$ after $v_{p-1}$ leads to a consistent
    configuration. Finally, from Property~\ref{updateorder}, $\forall
    r \in [p,k] : \upd(C_i,Unv_1\cdots v_r) = \upd(C_i,Uv_1\cdots nv_p\cdots v_r)$,
    so every node after $v_p$ can be updated in $V''$, since it could
    be updated in $T$. 
Let $C_1=\upd(C_i,Unv_1v_2\cdots v_{p-1})$ be the configuration before
updating $v_p$ in $T$. To connect $n$ to $H_1$, the update of $v_p$ when
the network is in configuration $C_1$ will add a $C_f$-only edge upstream
to $n$ and create a $C_f$ path between $v_p$ and $n$. For consistency
with this $C_f$-only edge, in $C_1$, all downstream maximal paths from
$n$ are in $\maxpaths(n,C_f)$. In $C_1$, $n$ satisfies the Type A
downstream condition. 
$C_1=\upd(C_i,Unv_1v_2v_3\cdots v_{p-1}) =\upd(C_i,Uv_1v_2\cdots v_{p-1}n)$,
so in $V''$, $n$ satisfies the downstream condition and is a Type A
valid node. If $V''$ starts with a disconnected invalid node, we
repeat this process until we find $V'''=n'V'$ where $n'$ is a valid
node. We are guaranteed to find $V'''$, because we continue changing
invalid disconnected nodes to valid nodes, and there can be only a
finite number of invalid disconnected nodes in $T$. 
\end{proof}
}

\makeproof{\thmEquivMain}{th:equiv_main}{theoremt}{Theorem}
{
If a correct sequence of updates exists, then a careful
sequence also exists. 
}
{
Let $Q=s_1s_2\cdots s_n$ be a correct update sequence. Let r be the first
index s.t. $\forall i<r:s_i$ is valid and $s_r$ is invalid. Then using
Lemma~\ref{lemma:equiv_main}, there is a sequence
$Q'=s_1's_2'\cdots s_n'$ s.t. $\forall i \le r:s_i'$ is valid. Using this
argument for every index up to $n$, we can find a $Q''$ s.t. $Q''$ is a
careful sequence.
}{}{}

\begin{algorithm}[t]
\DontPrintSemicolon
\footnotesize
\SetKwInOut{Input}{Input}
\Input{Set of all nodes $N$, Initial configuration $C_i$, Final configuration $C_f$}
\KwResult{An consistent order of node updates $R$, Updates before which there are waits $R_w$}
$R = R_w = P_0 \gets \emptyset; k \gets 1$\tcp*{initialize $R$, $R_w$, $P_0$ and $k$}
$C_c \gets C_i$\tcp*{$C_c$ starts with the initial value of $C_i$}
\While(\tcp*[f]{stop when $C_c$ and $C_f$ are equal}){$C_c \neq C_f$}{
	$U \gets \{s \mid s \in N \land((Y_a(s) \land Z_a(s)) \lor (Y_b(s) \land Z_b(s)) \;\lor $ \newline\hphantom{$U_1 \gets $}$(Y_c(s) \land Z_c(s)) \lor (Y_d(s) \land Z_d(s)) \lor (Y_e(s) \land Z_e(s)))\}$\label{alg:line:valid}\tcp*{valid nodes}
	\lIf(\tcp*[f]{no consistent order of updates exists}){$U = \emptyset$}{
		EXIT\label{alg:line:fail}
	}
	$s = PickAndWait()$\label{alg:line:pick}\tcp*{by default, use Algorithm~\ref{alg:pickdef}}
	$C_c \gets (C_c \setminus \out(s,C_i)) \cup \out(s,C_f)$\label{alg:line:remedge}\tcp*{update $C_c$}
	$N \gets N - \{s\}$\tcp*{remove updated nodes from node list}
}
\Return $(R, R_w)$
\caption{$\mathit{OrderUpdate}$}
\label{alg:orderupdate}
\end{algorithm}
\begin{algorithm}[t]
\DontPrintSemicolon
\footnotesize
$s = Pick(U)$\tcp*{pick any valid node}
$R_w \gets R_w.s$\tcp*{by default, there is a wait after every update}
$R \gets R.s$\tcp*{append $s$ to the end of result $R$}
\caption{$\mathit{SequentialPickAndWait}$}
\label{alg:pickdef}
\end{algorithm}
\begin{algorithm}[t]
\DontPrintSemicolon
\footnotesize
\If(\tcp*[f]{we do not need a wait before first node}){$k=1$}{
  $P_0 \gets U$\tcp*{all nodes initially valid are $P_0$}
}
\If(\tcp*[f]{we have to pick a lower priority node}){$P_0 = \emptyset$}{
  $P_0 \gets U$\tcp*{all nodes in $U$ become $P_0$ after waiting.}
  $s = Pick(P_0)$; $R \gets R.s$; $R_w \gets R_w.s$; $k \gets k +1$;\tcp*{pick $P_0$ node, append $s$ to result $R$, add wait, increment number of rounds $k$}
}
\Else{
  $s = Pick(P_0)$; $R \gets R.s$\tcp*{pick any $P_0$ node, add $s$ to result $R$}
}
\caption{$\mathit{OptimalPickAndWait}$}
\label{alg:pickandwait}
\end{algorithm}

\subsection{Completeness of the OrderUpdate Algorithm}
The OrderUpdate Algorithm (with the 
SequentialPickAndWait subroutine) is complete, i.e., if there
exists any correct sequence, we find one. 
We can observe that if two
nodes $a$ and $b$ are both valid in configuration $C_c$, then
$\upd(C_c,ab)$ and $\upd(C_c,ba)$ are both consistent
configurations. This property holds for any number of nodes and for all
{\em careful} sequences, but not for all {\em correct} sequences. We prove
this behavior in the following lemma, which is the key to observe
completeness of OrderUpdate Algorithm. 

\makeproof{\lemmacompleteness}{lemma:completeness}{lemmat}{Lemma}
{
If $T = UVnY$ is a careful sequence, and $\valid(\upd(C_i,U),n)$, then
$T' = UnVY$ is also careful.
}
{
\begin{figure}
\centering
\begin{minipage}{.5\textwidth}
  \centering
  \tikzset{ 
	VertexStyle/.append style = { minimum size = 12pt, inner sep = 0pt, font =\tiny, color=black }
	}
  \begin{tikzpicture}
  \SetGraphColor{white}{blue}{magenta}
  \Vertex[L=$H_1$]{H1}
  \EA[L=$v_r$,unit=1.25](H1){B}
  \EA[L=$ $,unit=1.25](B){D}
  \NOEA[L=$n$,unit=1.25](H1){A}
  \SOEA[L=$ $,unit=1.25](H1){C}

  \draw[dashed,-triangle 45](H1.20) -- (B.160);
  \draw[-triangle 45](H1.-20) -- (B.-160);
  
  \draw[dashed,-triangle 45](B.20) -- (D.160);
  \draw[-triangle 45](B.-20) -- (D.-160);
  
  \draw[dashed,-triangle 45](H1) -- (A);
  \draw[dashed,-triangle 45](A) -- (B);
  \draw[dashed,-triangle 45](B) -- (C);
  
  \end{tikzpicture}
  \caption{Lemma~\ref{lemma:completeness} Case 1.}
  \label{fig:case1}
\end{minipage}%
\begin{minipage}{.5\textwidth}
\centering
  \centering
  \tikzset{ 
	VertexStyle/.append style = { minimum size = 12pt, inner sep = 0pt, font =\tiny, color=black }
	}
  \begin{tikzpicture}
  \SetGraphColor{white}{blue}{magenta}
  \Vertex[L=$H_1$]{H1}
  \EA[L=$n$,unit=1.25](H1){B}
  \EA[L=$ $,unit=1.25](B){D}
  \NOEA[L=$v_r$,unit=1.25](H1){A}
  \SOEA[L=$ $,unit=1.25](H1){C}
  \EA[L=$ $,unit=1.25](C){E}

  \draw[dashed,-triangle 45](H1.20) -- (B.160);
  \draw[-triangle 45](H1.-20) -- (B.-160);
  
  \draw[dashed,-triangle 45](B.20) -- (D.160);
  \draw[-triangle 45](B.-20) -- (D.-160);
  
  \draw[dashed,-triangle 45](C) -- (E);
  \draw[dashed,-triangle 45](B) -- (C);
  \draw[dashed,-triangle 45](H1) -- (A);
  \draw[dashed,-triangle 45](A) -- (B);

  \end{tikzpicture}
  \caption{Lemma~\ref{lemma:completeness} Case 2.}
  \label{fig:case2}
\end{minipage}%
\end{figure}
Let $V=v_1\cdots v_k$, then $\forall r \in [1,k]:
C_r=upd(C_i,Uv_1\cdots v_{r})$ and $C_r'=upd(C_i,Unv_1\cdots v_{r})$ are the
configurations after updating $v_r$ in $T$ and $T'$ respectively. 
We will argue for each node $v_r$ in $V$, that $C_r'$ is
consistent. It is trivial to see that $\forall p \in \maxpaths(H_1,C_r') \cap
\maxpaths(H_1,C_r): p$ is consistent.
So, we only need to prove that $\forall p \in \maxpaths(H_1,C_r')
\setminus \maxpaths(H_1,C_r): p$ is consistent. Each $v_r$ can be classified into one of several
types based on maximal paths in $\maxpaths(H_1,C_r') 
\setminus \maxpaths(H_1,C_r)$.
\begin{compactitem}
  \item $\text{Case~1:}~\exists p \in \paths(H_1,v_r,C_r'): p \not\in
    \paths(H_1,v_r,C_r) \land \lnot(\exists p \in \maxpaths(v_r,C_r'):
    p \not\in \maxpaths(v_r,C_r))$. See Figure~\ref{fig:case1}. There are upstream paths to $v_r$
    in $C_r'$ not present in $C_r$. No downstream maximal
    paths from $v_r$ were added in $C_r'$. Consider sets of paths
    in $C_r'$ touching $v_r$: 
    \begin{compactenum}
      \item{$\up$} $= \paths(H_1,v_r,C_r)$ -- set of upstream paths to
        $v_r$ in $C_r$. 
      \item{$\up'$} $= \{p \mid p \in \paths(H_1,v_r,C_r'): p \not\in
        C_r\}$ -- set of upstream paths to $v_r$ in $C_r'$ which are
        not in $C_r$. Updating a node adds $C_f$-only edge(s) to the
        network, so for any path $p$ containing any of these edges $p
        \in C_f \land p \not\in C_i$. Hence, $\forall p \in up': p \in
        C_f \land p \not\in C_i$. 
      \item{$\down$} $= \maxpaths(v_r,C_r') \subseteq \maxpaths(v_r,C_r')$ --
        set of downstream paths from $v_r$ in $C_r'$.
    \end{compactenum}
    Let us define the $\cdot$ operator on two sets of paths $S$ and
    $S'$. We use $S \cdot S'$ to mean the set of all paths formed by the
    concatenation of any two paths $p \in S$ and $p' \in S'$ s.t. $p'$
    starts at the same node where $p$ ends. 
    All paths in $\maxpaths(H_1,C_r) \supseteq \up\cdot\down$ are
    consistent. 
    \begin{equation}
    \label{c11}
    \forall p \in (\up\cdot\down) : p \in \maxpaths(H_1,C_i) \lor p
    \in \maxpaths(H_1,C_f) 
    \end{equation}
    Let us partition $\down$ into $\down_1$ and $\down_2$. The set $\down_1$
    contains downstream maximal paths from $v_r$ that existed in
    $C_{r-1}'$ and $\down_2 = \down \setminus \down_1$. We inductively
    assume $\maxpaths(H_1,C_{r-1}') \allowbreak \supseteq \allowbreak (\up \cup \up')
    \cdot \down_1$ is consistent. 
    \begin{equation}
    \label{c12}
    \forall p \in (\up'\cdot\down_1) :  p \in \maxpaths(H_1,C_i) \lor
    p \in \maxpaths(H_1,C_f) 
    \end{equation}
    We know $down_2 \in \maxpaths(v_r,C_f)$ since they were added by some
    update. Paths in $up'$ are $C_f$ paths. 
    \begin{equation}
    \label{c13}
    \forall p \in (\up'\cdot\down_2) :  p \in \maxpaths(H_1,C_f)
    \end{equation}
    From Equations~\ref{c11},~\ref{c12}, and~\ref{c13}, we conclude that:
    $$\forall p \in ((\up' \cup \up) \cdot \down) :  p \in
    \maxpaths(H_1,C_i) \lor p \in \maxpaths(H_1,C_f)$$ 
    Thus, $C_r'$ is consistent, since all maximal paths from $H_1$ that
    touch $v_r$ are consistent.
  \item $\text{Case~2:}~\lnot(\exists p \in \paths(H_1,v_r,C_r'): p \not\in
    \paths(H_1,v_r,C_r)) \land (\exists p \in \maxpaths(v_r,C_r'): p
    \not\in \maxpaths(v_r,C_r))$. See Figure~\ref{fig:case2}. There are downstream maximal paths
    from $v_r$ in $C_r'$ which were not present in $C_r$. No upstream
    paths to $v_r$ were added. Similar to the previous case, let us
    define three sets of paths in $C_r'$ that touch $v_r$: 
    \begin{compactenum}
      \item{$\down$} $= \maxpaths(v_r,C_r)$ -- set of downstream paths in $C_r$.
      \item{$\down'$} $= \{p \mid p \in \maxpaths(v_r,C_r'): p \not\in
        C_r\}$ -- set of downstream maximal paths from $v_r$ 
        not present in $C_r$ but are present in $C_r'$. Similar to
        $up'$ in $\mathit{Case~1}$, $\forall p \in \down': p \in C_f \land p
        \not\in C_i$. 
      \item{$up$} $= \paths(H_1,v_r,C_r') \subseteq \paths(H_1,v_r,C_r)$ -- set
        of upstream paths to $v_r$ in $C_r$.
    \end{compactenum}
    We know that $\maxpaths(H_1,C_r) \supseteq \up \cdot \down$ is a consistent
    configuration, so Equation~\ref{c11} holds. 
    Since updating $n$ made changes to the downstream paths from
    $v_r$, node $n$ lies on a downstream maximal path from $v_r$. Also,  
    $\forall p \in \paths(v_r,n,C_r') : p \in C_f$, because if $v_r$
    and $n$ are connected by a path only in $C_i$, then updating $n$
    before $v_r$ in $T'$ would not be able to add $C_f$ paths to
    $C_r'$ (due to consistency reasons). This leads to one of two
    cases: 
    \begin{compactitem}
      \item $\forall p \in \paths(v_r,n,C_r') : p \in C_f \land p \in
        C_i$, i.e. $v_r$ and $n$ were connected from the start. Since all
        paths in $\down'$ touch $n$ ($C_r$ and $C_r'$ were different
        because $n$ was updated in $C_r'$), the update of $v_r$ in
        $C_{r-1}$ does not add any paths to
        $\down'$. $\forall p \in \down': p \in
        \maxpaths(v_r,C_{r-1})$. Configuration $C_{r-1}$ is
        consistent and $\maxpaths(H_1,C_{r-1}) \supseteq up \cdot
        down'$, $\forall p \in (up \cdot down'): p$ was consistent. 
      \item $\exists p \in \paths(v_r,n,C_r') : p \in C_f \land p
        \not\in C_i$, i.e. $v_r$ and $n$ are connected by a $C_f$-only
        path. This path existed in $C_r$, so paths in $up$ can exist
        in a consistent configuration with downstream maximal
        $C_f$-only paths. Paths in $up$ can exists with paths in
        $down'$ in a consistent configuration.     
    \end{compactitem}
    \begin{equation}
    \label{c21}
    \forall p \in (\up \cup \down'): p \in \maxpaths(H_1,C_i) \lor p
    \in \maxpaths(H_1,C_f) 
    \end{equation}
    From Equation~\ref{c11} and Equation~\ref{c21}:
    $\forall p \in \maxpaths(H_1, C_r' = \up \cup \down \cup \down')
    :  p \in \maxpaths(H_1,C_i) \lor p \in \maxpaths(H_1,C_f)$, meaning 
    $C_r'$ is a consistent state and $v_r$ can be updated.
  \item $\text{Case~3:}~\exists p \in \maxpaths(v_r,C_r'): p \not\in
    \maxpaths(v_r,C_r) \land \exists p \in \paths(H_1,v_r,C_r'): p
    \not\in C_r$, i.e. updating $n$ added some upstream paths to $v_r$ and
    some downstream maximal paths from $v_r$. So, $n$ was both
    upstream to $v_r$ and downstream from $v_r$. This case is not
    possible because updating $n$ does not add any cycles to the
    network. 
  \item $\text{Case~4:}~\not\exists p \in \maxpaths(v_r,C_r'): p \not\in
    \maxpaths(v_r,C_r) \land \not\exists p \in \paths(H_1,v_r,C_r'): p
    \not\in C_r$, i.e. there has been no change in upstream and downstream
    paths. So, $C_r'$ is a consistent state.%
\end{compactitem}
We have seen that every $v_r$ in the sequence $V$ can be updated in
$T'$. Also, using Property~\ref{updateorder}, $upd(C_i,UnV) =
upd(C_i,UVn)$, nodes in $Y$ can be updated in sequence. Hence we
showed that if $T = UVnY$ is a correct careful sequence, $T' = UnVY$
is a correct careful sequence.
}{}{}

Lemma~\ref{lemma:completeness} shows that if there are multiple valid
nodes in some configuration $C$, then these nodes can be updated in
any order. This is because once a node becomes valid, it does not
become invalid. This is why we introduced careful
sequences because this lemma is not true for arbitrary correct
sequences. Using this lemma, we can prove the completeness of
Algorithm~\ref{alg:orderupdate} (with the Algorithm~\ref{alg:pickdef}
subroutine).  

\makeproof{\thmcompleteness}{th:completeness}{theoremt}{Theorem}
{
Algorithm~\ref{alg:orderupdate}, using subroutine
Algorithm~\ref{alg:pickdef}, generates a correct order of updates
$R$ if there exists one, or fails (in Line~\ref{alg:line:fail}) if such an order does not exist. 
}
{
We proved the correctness of Algorithm~\ref{alg:orderupdate}, using
subroutine SequentialPickAndWait, in
Theorem~\ref{th:correctness}. So we know that if it generates an
order of updates, it is correct. 

Let us consider the case where a correct sequence of updates exists
but Algorithm~\ref{alg:orderupdate} fails. Let $Q_{\mathit{careless}}$
be the correct sequence of updates, and $Q_{alg}=a_1a_2\cdots a_k$ be
the sequence of nodes updated by Algorithm~\ref{alg:orderupdate}
before it fails. Using Theorem~\ref{th:equiv_main}, let
$Q_{\mathit{careful}}=s_1s_2\cdots s_n$ be a careful sequence. Let r be
the first index s.t. $\forall i<r: s_i=a_i \land s_r \neq a_i$. If
$r<k$, then using Lemma~\ref{lemma:completeness}, there is another
careful sequence $Q_{\mathit{careful}}' = s_1's_2'\cdots s_n'$
s.t. $\forall i \leq r: s_i'=a_i$. Using this argument for every index
up to $k$, we can find a correct careful sequence
$Q_{\mathit{careful}}''$ s.t. $Q_{alg}$ is a prefix sequence of
$Q_{\mathit{careful}}''$. So, there is a correct node after nodes in
$Q_{alg}$ were updated and Algorithm~\ref{alg:orderupdate} could not
have failed. Therefore, if Algorithm~\ref{alg:orderupdate} fails, then
no correct sequence of updates exists. 
}{}{}
\paragraph{Running Time.}  
Let $|V|$ be the number of nodes and $|E|$ be the number of edges in
$G$.  In each iteration of its outer loop,
Algorithm~\ref{alg:orderupdate} using $\mathit{SequentialPickAndWait}$ (Algorithm~\ref{alg:pickdef}) as a
subroutine, makes a list of valid nodes and picks one to update.  The
set of valid nodes $U$ in Line~\ref{alg:line:valid} can be found using a graph search on $C_c$ for each node, which takes $O(|V|(|V| + |E|))$ steps.
The loop runs $|V|$ times and updates each node, so the overall runtime is $O(|V|^2(|V|+|E|))$. This analysis relies on the fact that the graph search is implemented in a way that goes through each edge and node a constant number of times. Once a node has been visited, it is marked $F$, $I$, or $B$, based on whether the maximal paths downstream from it are maximal paths starting from it in $C_i$, $C_f$, or both. This would avoid visiting the node (and its outgoing edges) again. 

\section{Optimal OrderUpdate Algorithm}
\label{sec:waits}

Thus far, we solved the consistent order update problem by generating
a consistent sequence with only singleton sets. This corresponds to
requiring a wait at every step of the update sequence, which does not
allow any parallelism. However, we have seen in
Section~\ref{sec:overview} that some nodes can be updated in
parallel. In Section~\ref{sec:model}, we defined when a wait is needed
in the sequence of updates. In this section, we provide a sequence of
updates where there is a wait if and only if it is needed, solving the
optimal version of the problem. We use
Algorithm~\ref{alg:orderupdate}, but replace the subroutine
$\mathit{SequentialPickAndWait}$ (Algorithm~\ref{alg:pickdef}) with
$\mathit{OptimalPickAndWait}$ (Algorithm~\ref{alg:pickandwait}). The
algorithm returns a solution for the optimal consistent update problem
in the following format.

\paragraph{Correct Waited Sequence.}  
A correct waited sequence of updates is a tuple $(T,W)$ of node
sequences without repetition, where $W$ is a subsequence of $T$ and
$(T,W)=(t_1t_2\cdots t_{|N|},w_1w_2\cdots w_{k-1})$, such that a
consistent update sequence $S_1S_2\cdots S_k$ can be formed by taking
$S_1=\{t_1,\cdots,t_m\}$ where $t_{m_1}=w_1$, $\forall i \in (1,k):
S_i =\{t_{l_i},\cdots,t_{m_i}\}$ where $t_{l_i}=w_{i-1}$ and
$t_{m_i}=w_i$, and $S_k=\{t_{l_k},\cdots,t_{|N|}\}$ where
$t_{l_k}=w_{k-1}$.

Intuitively, $T$ specifies a correct sequence of updates, with some
waits, while $W$ specifies the nodes, immediately
before which a wait is placed.  If we simply group the nodes between
$i$-th and $(i+1)$-st waits into a set $S_{i+1}$ we obtain the
consistent update sequence of Section~\ref{sec:model}. Considering
solutions to the problem in the form of a sequence of nodes and waits
simplifies the arguments we use to prove correctness and optimality.

\paragraph{Minimal Correct Waited Sequence.}
A {\em minimal correct waited sequence} is a correct waited sequence
$(T,W)$ such that $|W|$ is minimal.

Since we always pick valid nodes, we need to prove that if there
exists a minimal correct waited sequence, then there exists a minimal
correct waited sequence that updates only valid nodes.

\paragraph{Careful Waited Sequence.}
A {\em careful waited sequence} of updates
$(T,W)=(t_1t_2\cdots t_{|N|},w_1w_2\cdots w_{k-1})$ is a correct waited
sequence s.t. $\forall j \in [1,|N|]: \valid(\upd(C_i,t_1\cdots t_{j-1}),t_j)$
A {\em minimal careful waited sequence} is a careful waited sequence $(T,W)$ s.t. $|W|$ is minimal. 
We prove the following for such sequences.

\extendedversion{
\begin{lemmat}
\label{lemma:equiv_wait}
Let $Z=(UnV,W=w_1\cdots w_k)$ be a correct waited sequence where $n$ is an
invalid disconnected node, then $\exists Z'=(Un'V',W')$, a correct
waited sequence in which $n'$ is a valid node, and $V'$ is a sequence
s.t. $n'V'=\pi(nV)$ and $|W| = |W'|$. 
\end{lemmat}
\begin{proof}
To prove Lemma~\ref{lemma:equiv_wait}, we use the same transformation
as Lemma~\ref{lemma:equiv_main} and update $n$ immediately before $v_p$, the
node that connects it to the network, in a waited sequence $Z'=
(V'',W')$, where $V''=Uv_1v_2\cdots v_{p-1}nv_p\cdots v_k$, and prove that
$|W| = |W'|$. 

Let us consider the case where there was no wait before $n$ in $Z$, i.e.
$n$ was not in sequence $W$. 
For each node $s \neq n$, let $C_s$ and $C_s'$ be configurations after
updating $s$ in $Z$ and $Z'$ respectively. For any node $s \neq n$,
let $r$ be the latest node updated before $s$ in $Z$ which had a wait
before it ($r$ is the last node in $W$). Let us form two unions 
$S = C_r \cup\cdots \cup C_s$ and $S' =
C_r' \cup \cdots \cup C_s'$, consisting of unions of all intermediate
configurations between $r$ and $s$ in $Z$ and $Z'$. 

\begin{compactitem}
  \item Node $s$ was updated before $n$ in $Z$. In this case $S=S'$ as
    there was no change in updates before $n$ in $Z'$. Since $S=S'$,
    no wait is required before $s$ in $Z'$ if no wait was required in
    $Z$. 
  \item Node $s$ was updated between $n$ and $v_p$ in $Z$. In $Z'$, $n$
    was not updated. There are two subcases: 
  \begin{compactitem}
    \item Node $r$ was updated after $n$ in $Z$. For this subcase 
    $S'\setminus S=\out(n,C_i)\setminus\out(n,C_f)$. 
    However, since $n$ was
      disconnected in all configurations between $C_r$ and $C_s$,
      consistency of $S'$ is not affected by these edges, as there are
      no maximal paths from $H_1$ that go through $n$. Hence $S'$ is
      consistent if $S$ is consistent. 
    \item Node $r$ was updated before $n$ in $Z$. For this subcase,
    $S'$ had only edges from $\out(s,C_i)$. Additionally, $S$ had edges from both
    $out(s,C_i)$ and $\out(s,C_f)$. So, $S'\setminus S=\emptyset$.
    $S'$ is consistent if $S$ is consistent. 
  \end{compactitem}
  In both subcases, no additional waits are required before $s$ in
  $Z'$. 
  \item We have $s = v_p$, or $s$ was updated after $v_p$. There are again two
    subcases here: 
  \begin{compactitem}
    \item Node $r$ was updated before $v_p$ in $Z$. In this subcase, 
    $S' \setminus S = \out(s,C_i)\setminus\out(s,C_f)$. Let us consider
    $C_1=\upd(C_i,Uv_1\cdots v_{p-1}n)$ and $C_2=\upd(C_i,Uv_1\cdots v_{p-1}nv_p)$.
    Configuration $C_2$ adds a $C_i$ path $p$ from $v_p$ to $n$ which was not present in $C_1$.
    Since there was no wait between $n$ and $v_p$, $C_1 \cup C_2$ in consistent.
    So, because there was $C_f$ upstream path from $H_1$ to $n$ in $C_2$, $C_1$ had downstream 
    maximal paths from $n$ which were all in $C_f$. However, $C_1$ had paths in $out(n,C_i)$.
    This is only possible if $out(n,C_i) \subseteq out(n,C_f)$. So,
    $\out(s,C_i)\setminus\out(s,C_f) = \emptyset$ and $S= S'$. $S'$ is consistent if $S$ is consistent.
    \item We have $r=v_p$, or $r$ was updated after $v_p$ in $Z$. In this case,
      $S = S'$ because $\forall j>p :
      C_j = C_j'$. So, $S'$ is consistent if $S$ is consistent. 
  \end{compactitem}
\end{compactitem}
We argued for all $s \neq n$ that the waits do not move.
Now, let us argue for $n$.
Let $m$ be the latest node before $n$ s.t. for some $j$, $w_j=m$. Then
two cases are possible: 
\begin{compactitem}
  \item In $Z$, no node in the sequence $v_1\cdots v_{p-1}$ is in $W$. Let
    $C_m$ be the configuration before updating $m$ in $Z$.  
  Since there was no wait before $n$ in $Z$, we know that $S = C_m
  \cup \cdots \cup \upd(C_i,U) \cup \upd(C_i,Un)$ is consistent. We proved 
  that waits in $Z'$ for nodes $s\neq n$ are required at the same location as $Z$.
  So, $S' = C_m \cup \cdots \cup \upd(C_i,U)\cup
  \upd(C_i,Uv_1)\cup\cdots\cup\upd(C_i,Uv_1\cdots v_{p-1})$ is consistent. Let us
  consider $S'' = S' \cup \upd(C_i,Uv_1\cdots v_{p-1}n)$. If there were
  any inconsistent paths in $S''$, they were also a part of $S$ (since
  $n$ is not connected to $H_1$ in any configuration
  $\upd(C_i,Uv_1\cdots v_l)$ where $l<p$). So, there is no wait needed
  before $n$. 
  \item In $Z$, $\exists r \in [1,p)$ s.t. $v_{r}$ is in $W$. Let $q$
    be the greatest index for which $v_q$ satisfies this
    condition. Consider $S = \upd(C_i,Unv_1\cdots v_q)\cup \cdots \cup
    \upd(C_i,Unv_1\cdots v_{p-1})$ and $S'= \upd(C_i,Uv_1\cdots v_q)\cup
    \cdots \cup \upd(C_i,Uv_1\cdots v_{p-1}) \cup
    \upd(C_i,Uv_1\cdots v_{p-1}n)$. We proved 
  that waits in $Z'$ for nodes $s\neq n$ are required at the same location as $Z$. So, in $Z'$,
      $v_q$ was the latest node in $V$ before which there was a
      wait. Then, maximal paths from $H_1$ in both $S$ and $S'$ are
      the same, since $n$ was not connected to $H_1$ before $v_p$ is updated.
      So there is no wait needed before $n$.
\end{compactitem}

\noindent
In case there was a wait before $n$ in $Z$, we consider a sequence
$Z''=(Uv_1nv_2\cdots v_kY,W'')$. In $Z''$ there is a wait before $v_1$ but
not before $n$. This is because $n$ adds edges that are disconnected
from the network. So, there is no requirement for a wait between $v_1$
and $n$. For $Z''$, this becomes the case with no wait
before $n$. 
\end{proof}
}

\makeproof{\theoremequivwait}{th:equiv_wait}{theoremt}{Theorem}
{
If a minimal correct waited sequence exists, then a minimal careful sequence exists as well.
}{The proof uses Lemma~\ref{lemma:equiv_wait} and is similar to the proof of Theorem~\ref{th:equiv_main}.
}{}{}

\subsection{Condition for Waits}

\paragraph{Partial Careful Waited Sequence.}
Given careful waited sequence $Z=(T=t_1\cdots t_{|N|},W=w_1\cdots w_{k-1})$, a
partial careful waited sequence is $Z'=(T'=t_1\cdots t_r,W'=w_1\cdots w_s)$
such that $T'$ is a prefix of $T$ and $W'$ is a prefix of $W$.  The update
mechanism starts with a  
partial careful waited sequence with no nodes
and at every step, it adds a node in a way that ensures that the
obtained sequence is a partial careful waited sequence, i.e., it can
be extended to a careful waited sequence. 

\paragraph{Wait Condition.}
Let us define a function $\wait$ that
takes a partial careful waited sequence $S = (t_1t_2\cdots t_r,\allowbreak w_1w_2\cdots w_s)$ 
and node $n$ s.t. $\valid(C_i,Ut_1\cdots t_r)$ as an argument and returns 
$\true$ if there needs to be a wait before its update. 
It is defined as follows: 
$\wait(n,S)=\true$ iff node $\exists x \in [1,r]: \lnot\valid(\upd(C_i,t_1\cdots t_x),n) \land \lnot(\exists y \in [1,s], \exists z \in (x,r]: w_y = t_z)$.
In other words, in the partial careful waited sequence, there
must be a wait before updating a valid node $n$ if and only if it was
not valid until its dependencies were updated, and there was no wait
after their update. If this is true, then $n$ must be updated in a new
round, after a wait.

The following shows {\em completeness} of the wait condition,
i.e., if a wait is needed (as defined in
Section~\ref{sec:model}) after updating $S$ and
before updating $n$, then $\wait(n,S)$ is true. 

\makeproof{\lemmawaitcond}{lemma:waitcond}{lemmat}{Lemma}
{
If (1) $n$ is the node picked for update,
and (2) the partial careful waited sequence built before updating $n$ is
$S=(t_1t_2\cdots t_r,w_1w_2\cdots w_s)$, and (3) $w_s = t_y$ for some $y \in
[1,r]$, and (4) we define $\forall x \in
[1,r]:C_{t_x}=\upd(C_i,t_1\cdots t_x)$, and then $\wait(n,S)
\leftrightarrow C_{t_y} \cup \cdots \cup C_{t_r} \cup \upd(C_{t_r},n)$ is 
inconsistent.
}
{
Let us first prove that $\wait(n,S) \rightarrow C_{t_y} \cup \cdots \cup C_{t_r} \cup \upd(C_{t_r},n)$ is inconsistent.
For some $a > y$, let $C_{t_a}$ be the configuration of the network in which $n$ was invalid. We know $t_a$ was updated after $t_y$, so there was no wait between the update of $t_a$ and $t_r$. Updating $n$ in $C_{t_a}$ would lead to a inconsistent configuration $C_{t_a}'=\upd(C_{t_a},n)= (C_{t_a} \setminus \out(n,C_i)) \cup \out(n,C_f)$. Now, $C_{t_a} \setminus \out(n,C_i) \subset C_{t_a}$ and $\out(n,C_f) \subset \upd(C_{t_r},n)$. Therefore, $(C_{t_a} \setminus \out(b,C_i)) \cup \out(b,C_f) = C_{t_a}' \subseteq C_{t_a} \cup \upd(C_{t_r},n)$ . Therefore, if $\wait(n,S) = \false$, then $C_{t_a} \cup \upd(C_{t_r},n)$ cannot be consistent.

Now let us prove that $\lnot\wait(n,S) \rightarrow C_{t_y} \cup \cdots \cup C_{t_r} \cup \upd(C_{t_r},n)$ is consistent.
Since $\wait(n,S)=\false$, there are no waits between $t_y$ and $t_r$, $n$ was valid in every configuration reached between the updates of $t_y$ and $t_r$. This means $\forall z \in [y,r]: \upd(C_{t_z},n)$ is consistent.
Also $W = C_{t_y} \cup \cdots \cup C_{t_r}$ is consistent. Let us assume that $W' = C_{t_y} \cup \cdots \cup C_{t_r} \cup \upd(C_{t_r},n)$ is inconsistent. Then there is an inconsistent path in $W'$. However, since $W$ was consistent, this path was not from the union of configurations in $W$. 
So, this path had edges from set $W' \setminus W=\out(n,C_f) \setminus \out(n,C_i)$. Let us form the set $\mathit{add}(t_l)=\out(t_l,C_i)\setminus \out(t_l,C_f)$ which is the set of edges that are added to $C_{t_l}$ after its update. Consider these cases for each inconsistent path $p$ in $W'$.
\begin{compactitem}
  \item $p$ has no edges from $\mathit{add}(t_l)$ for any $t_l$, so $\upd(C_{t_y},n)$
  is inconsistent (impossible). %
  \item $p$ has edges from sets $\mathit{add}(t_{l_1}) \cup \cdots \cup\mathit{add}(t_{l_z})$ for some nodes $t_{l_1}\cdots t_{l_z}$ between $t_y$ and $t_r$ (inclusive), then let $t_{l_g}$ be the  node in set $\{t_{l_1},\cdots,t_{l_z}\}$ that occurs latest in sequence $t_1\cdots t_r$. So, $p$ existed in $\upd(C_{t_g},n)$. However, since we know that $\upd(C_{t_g},n)$ is consistent, this condition is also impossible.
\end{compactitem}
Using this argument for every inconsistent path in $W'$, we prove $W$ is consistent.
So, we have proved that the wait condition defined by function $\wait$ is complete.
}{}{}

\subsection{Algorithm for Optimal Consistent Order Updates}
We now present the $\mathit{OptimalPickAndWait}$
(Algorithm~\ref{alg:pickandwait}) subroutine, that minimizes the
number of waits, solving the optimal consistent update
problem. Our strategy for minimizing waits is to assign one of two
priorities to nodes: $P_0$ (higher priority) and $P_1$ (lower
priority). Let $S$ be a partial sequence. A node is in $P_0$ iff
$\lnot\wait(n,S)$, i.e. $P_0$ nodes do not require waiting before
update. A node is in $P_1$ iff $\wait(n,S)$, i.e. we must wait before
updating a $P_1$ node. We greedily update $P_0$ nodes first.

Correctness and optimality %
follow from
the correctness argument in the previous section, and from
Lemma~\ref{lemma:waitcond}.  Intuitively, updating a node in $P_0$
which does not need a wait allows the $P_1$ list to build up. This
means we need to place a single wait for as many $P_1$ nodes as
possible. When we place a wait in the partial careful waited sequence,
every valid node that was in $P_1$ moves to $P_0$. The last key
property needed for the following theorems is that 
once a node acquires priority $P_0$, it remains in $P_0$. 

\extendedversion{
\begin{lemmat}
\label{lemma:stays_valid}
If a node $n$ is valid in configuration $C$, then it is valid in configuration $\upd(C,n')$ for some valid node $n' \neq n$.
\end{lemmat}
\begin{proof}
For validity, we do not consider the waits. We can directly apply
Lemma~\ref{lemma:completeness}. If a node $n$ is valid in a correct
sequence $T=Unn'V$, then if $\valid(\upd(C_i,U),n')$, $T'=Un'nV$ is a
correct sequence, meaning $\valid(\upd(C_i,Un'),n)$. So, the update of
any other node does not affect the validity of $n$. 
\end{proof}

\begin{lemmat}
\label{lemma:stays_high}
If during the update, a node has priority $P_0$, it retains priority $P_0$ until it is updated.
\end{lemmat}
\begin{proof}
Node $n$ is a priority $P_0$ node when the partial careful waited sequence $Z = (t_1t_2\cdots t_r,w_1w_2\cdots w_s)$ has been built. If $n$ is updated after $t_r$, $\wait(n,Z)=\false$. However, from Lemma~\ref{lemma:stays_valid}, since $n$ stays valid in every configuration after the update of $t_r$, $\wait(n,Z)=\false$ no matter where $n$ is updated.
\end{proof}
}

\makeproof{\theoremwaitcorrectness}{th:wait:correctness}{theoremt}{Theorem}
{
Algorithm~\ref{alg:orderupdate} with Algorithm~\ref{alg:pickandwait}
as its subroutine on Line~\ref{alg:line:pick} produces a correct
waited sequence. 
}
{
Using Lemma~\ref{lemma:waitcond}, every node that is not valid at the
start is a priority $P_1$ node when it becomes valid. We pick $P_0$
nodes with higher priority, and do not wait before them. When $P_0 =
\emptyset$, we wait before we pick any node in $P_1$. By definition,
adding a wait changes the priority of all nodes in $P_1$ to $P_0$. From
Lemma~\ref{lemma:stays_high}, these nodes retain priority $P_0$ until they are
updated, showing that waits are correctly placed. 
}{}{}

\extendedversion{
We now prove that our greedy scheme is optimal. For this purpose, let us prove the following two lemmas:
\begin{lemmat}
\label{lemma:wait_1}
If $Z=(T,W) = (UVnY,w_1\cdots w_{k})$ is a careful waited sequence, and in $Z$, after updating nodes in $U$, $n \in P_0$, then $Z'=(T',W') = (UnVY,w'_1\cdots w'_k)$ is a careful waited sequence.
\end{lemmat}

\begin{proof}
From Lemma~\ref{lemma:completeness}, we know that $T'$ is a correct sequence. Here, in addition to $n$ being a valid node, $n$ is a Priority $P_0$ node. Since $n \in P_0$ after updating $U$, from Lemma~\ref{lemma:stays_high}, $s \in P_0$ in both $Z$ and $Z'$. So, $n$ does not get added in $W'$. The partial careful waited sequence consisting only of nodes in $U$ is the same for both $Z$ and $Z'$. Let us complete this sequence by arguing for each node $s$ in $VY$.
\begin{compactitem}
  \item{Case 1:} In $Z$, $s \in P_1$ ($s$ was in $W$). In $Z'$, we keep $s$ in $W'$. We do not add any nodes in $W'$ as compared with $W$.
  \item{Case 2:} In $Z$, $s \in P_0$ ($s$ was not in $W$). In $Z'$, $s \in P_0$. Since we have kept the waits at the same position as $Z$, if a wait was needed between any two nodes (excluding $n$) in $Z$, there is a wait in $Z'$. In $Z$, if $s$ became valid in some configuration $C$, then $s$ is also valid in $\upd(C,n)$ (Lemma~\ref{lemma:stays_valid}). A wait is needed before updating $s$ in $Z'$, if it was needed in $Z$.
\end{compactitem}
Hence we proved that $Z'$ is a careful waited sequence with $|W| = |W'|$.
\end{proof}

\begin{lemmat}
\label{lemma:wait_2}
If $Z=(T,W) = (UVnY,w_1\cdots w_{k})$ is a careful waited sequence, and in $Z$, after updating nodes in $U$, $P_0 = \emptyset \land n \in P_1$, then $Z'=(T',W') = (UnVY,w'_1\cdots w'_k)$ is a careful waited sequence.
\end{lemmat}
\begin{proof}
Similar to Lemma~\ref{lemma:completeness}, the partial careful waited sequence consisting only of nodes in $U$ is the same for both $Z$ and $Z'$. Let $V=v_1\cdots v_g$. Then since $P_0 = \emptyset$ after updating $U$, $v_1$ is in $W$. To construct $Z'$, let us swap $n$ for $v_1$ in $W'$. After this wait, $v_1 \in P_0$, so we do not need to add $v_1$ to $W'$. Then for all nodes $s$ in $v_2\cdots v_gY$, we argue in the same way as in Lemma~\ref{lemma:completeness}, and prove that $Z'$ is a careful waited sequence with $|W| = |W'|$. 
\end{proof}
}

\makeproof{\theoremwaitoptimality}{th:wait}{theoremt}{Theorem}
{
Algorithm~\ref{alg:orderupdate} with Algorithm~\ref{alg:pickandwait}
as its subroutine on Line~\ref{alg:line:pick} produces a correct
and optimal waited sequence of updates, if there exists a correct
waited sequence of updates. 
}
{
We have seen the correctness and completeness of
Algorithm~\ref{alg:orderupdate}. We also proved the correctness of our
approach for minimizing waits (Theorem~\ref{th:wait:correctness}). We
will now prove the optimality of Algorithm~\ref{alg:orderupdate} with
the Algorithm~\ref{alg:pickandwait} modification. Let
$Q_{careless}=(T_{careless},W_{careless})$ be an minimal correct
waited sequence, and $Q_{alg}=(a_1a_2\cdots a_n,b_1\cdots b_{n'})$ be the
sequence generated by Algorithm~\ref{alg:orderupdate} with
Algorithm~\ref{alg:pickandwait} as its subroutine. Using Lemma~\ref{lemma:equiv_wait},
we know there is a minimal careful waited sequence
$Q_{careful}=(s_1s_2\cdots s_n,w_1\cdots w_k)$. Let r be the first index
s.t. $\forall i<r: s_i=a_i \land s_r \neq a_r$. In $Q_{alg}$, if $a_r
\in P_0$, then by Lemma~\ref{lemma:wait_1}, we can generate a careful
sequence $Q'=(s'_1s'_2\cdots s'_n,w'_1\cdots w'_k)$ s.t. $\forall i \leq r:
s'_i=a_i$. In $Q_{alg}$, if $a_r \in P_1$, then from
Algorithm~\ref{alg:pickandwait} we know that $a_r$ was picked because
$P_0 = \emptyset$ after updating nodes
$s_1s_2\cdots s_{r-1}=a_1a_2\cdots a_{r-1}$. By Lemma~\ref{lemma:wait_2}, we
can again generate a minimal careful waited sequence
$Q'=(s'_1s'_2\cdots s'_n,w'_1\cdots w'_k)$ s.t. $\forall i \leq r:
s'_i=a_i$. Using this argument for every index from $i$ to $n$, we can
find a minimal careful waited sequence
$Q''=(s''_1s''_2\cdots s''_n,w''_1\cdots w''_k)$ s.t $\forall i:
s''_i=a_i$. Now since $\forall i: s''_i=a_i$, and our wait condition
is complete (Lemma~\ref{lemma:waitcond}), so $n'=k$. 
}{}{}

\optnewpage

\paragraph{Running Time.}  
The OrderUpdate Algorithm with the $\mathit{OptimalPickAndWait}$ subroutine has the same time complexity
that it had with the $\mathit{SequentialPickAndWait}$ subroutine.
The $\mathit{OptimalPickAndWait}$ subroutine introduces a priority-based node selection mechanism---%
after every wait, it simply moves nodes from the valid set $U$ to the higher priority list $P_0$,
which requires only $O(|N|)$ additional steps in each iteration.
\section{Discussion}
\label{sec:discussion}

\paragraph{Multiple hosts and sinks.}
We can extend our single-source approach to a network with multiple sources $H_A,H_B,H_C,\cdots$.
To do this, we assume that there is a master source $H_1$, and every actual source is connected
to $H_1$, as shown in Figure \ref{fig:multsrc}. This approach works because we update every node
only once, meaning we cannot artificially disable and then re-enable
some sources and keep others. 

\paragraph{Multiple packet types.}
Our approach can be applied in contexts where there are multiple (discrete) packet types,
as long as each forwarding rule matches on a {\em single} packet type---in this case, we simply compute an
update for each packet type, and perform these (rule-granularity) updates independently.
In the more realistic case with {\em symbolic} forwarding rules
(i.e., matching based on {\em first-order formulae over packet header
  fields}), deciding whether a consistent update exists is
\textsc{co-np}-hard.  Specifically, there is a reduction from SAT to
this problem.  In this case, we can consider each edge in a
configuration as being labeled by a formula, and only packets whose
header fields satisfy this formula can be forwarded along that edge.
To show the reduction, we consider a double diamond
(Figure~\ref{fig:doubled_multiple}) with one edge labelled by such a
formula $\varphi$, and all other edges labelled with {\em true}
($\top$).  We have already seen that a consistent update for this
double diamond example is not possible in the situation where packets
(of any type) can flow along all of the edges, so we can see that {\em
  there exists a consistent update if and only if $\varphi$ is
  unsatisfiable}.
This completes the reduction.

\begin{figure}[t]
\centering
\begin{minipage}{.35\textwidth}
  \centering
  \tikzset{ 
	VertexStyle/.append style = { minimum size = 12pt, inner sep = 0pt, font =\tiny, color=black }
	}
  \begin{tikzpicture}
  \SetGraphColor{white}{blue}{magenta}
  \Vertex[L=$H_1$]{H1}
  \EA[L=$H_B$,unit=1.1](H1){B}
  \NOEA[L=$H_A$,unit=1.1](H1){A}
  \SOEA[L=$H_C$,unit=1.1](H1){C}
  \draw[dashed,-triangle 45](H1.north) -- (A.west);
  \draw[-triangle 45](H1.north east) -- (A.south west);
  
  \draw[dashed,-triangle 45](H1.south east) -- (C.north west);
  \draw[-triangle 45](H1.south) -- (C.west);
  
  \draw[dashed,-triangle 45](H1.north east) -- (B.north west);
  \draw[-triangle 45](H1.south east) -- (B.south west);
  
  \draw[dashed,-triangle 45] (B) -- (5:2cm);
  \draw[-triangle 45] (B) -- (-5:2cm);

  \draw[dashed,-triangle 45] (A) -- (35:2.2cm);
  \draw[-triangle 45] (A) -- (25:2.1cm);

  \draw[dashed,-triangle 45] (C) -- (-25:2.1cm);
  \draw[-triangle 45] (C) -- (-35:2.2cm);

  \end{tikzpicture}
  \caption{Multiple sources.}
  \label{fig:multsrc}
\end{minipage}%
\begin{minipage}{.55\textwidth}
  \centering
  \tikzset{ 
	VertexStyle/.append style = { minimum size = 12pt, inner sep = 0pt, font =\tiny }
	}
  \begin{tikzpicture}
  \Vertex[L=$H_1$]{H1}
  \NOEA[unit=1.12](H1){A}
  \SOEA[unit=1.12](H1){B}
  \NOEA[unit=1.12](B){C}
  \NOEA[unit=1.12](C){D}
  \SOEA[unit=1.12](C){E}
  \NOEA[unit=1.12,L=$H_2$](E){H2}

  \draw[-triangle 45](H1.south east) -- node[midway,fill=white,font=\small] {$\top$} (B.north west);
  \draw[dashed,-triangle 45](H1.north east) -- node[midway,fill=white,font=\small] {$\varphi$} (A.south west);
  \draw[dashed,-triangle 45](A.south east) -- node[midway,fill=white,font=\small] {$\top$} (C.north west);
  \draw[-triangle 45](B.north east) -- node[midway,fill=white,font=\small] {$\top$} (C.south west);
  \draw[dashed,-triangle 45](C.north east) -- node[midway,fill=white,font=\small] {$\top$} (D.south west);
  \draw[-triangle 45](C.south east) -- node[midway,fill=white,font=\small] {$\top$} (E.north west);
  \draw[dashed,-triangle 45](D.south east) -- node[midway,fill=white,font=\small] {$\top$} (H2.north west);
  \draw[-triangle 45](E.north east) -- node[midway,fill=white,font=\small] {$\top$} (H2.south west);
  \end{tikzpicture}
  \caption{Double diamond case with symbolic forwarding rules.}
  \label{fig:doubled_multiple}
\end{minipage}%
\end{figure}

\section{Related Work}
\label{sec:relwork}

\paragraph{Consistency.}
Our core problem is motivated by earlier work by
\citet{Reitblatt:2012:ANU:2342356.2342427} that proposed {\em per-packet consistency}
and provided basic update mechanisms. %

\paragraph{Exponential Search-Based Network Update Algorithms.}

There are various approaches for producing a sequence of switch updates guaranteed to respect certain
path-based consistency properties (e.g., properties representable using temporal logic, etc.).
For example, \citet{McClurg:2015:ESN:2737924.2737980} use counter-example guided search and incremental LTL model checking,
FLIP \cite{VC16} uses integer linear programming,
and CCG \cite{189032} uses custom reachability-based graph algorithms.
{Other works such as Dionysus \cite{Jin:2014:DSN:2619239.2626307},
zUpdate \cite{Liu:2013:ZUD:2486001.2486005}, and
\citet{luo2016arrange}, seek to perform updates with respect to quantitative properties.}

\paragraph{Complexity results.}
\citet{MahajanW13}
introduce dependency-graphs for network updates,
and propose properties which could be addressed via this general approach.
They show how to handle one of the properties ({\em loop-freedom}) in a minimal way.
\citet{yuan2014generating} detail general algorithms for building dependency
graphs and using these graphs to perform a consistent update.
\citet{forsterconsistent} extend \cite{MahajanW13}, and
show that for {\em blackhole-freedom}, computing an update with a minimal
number of rounds is \textsc{np}-hard (when memory limits are assumed on switches).
They also show \textsc{np}-hardness results for rule-granular loop-free updates with maximal parallelism.
Per-packet consistency in our problem is stronger than loop freedom and blackhole freedom,
but we only consider solutions where each switch is updated {\em once}, and where a switch update
swaps the entire old forwarding table with the new one simultaneously.

{\citet{forster2016power} examine loop-freedom,
showing that maximizing the number for forwarding rules updated simultaneously is \textsc{np}-hard.}
\citet{Ludwig:2015:SLN:2767386.2767412} show how to minimize number of update
rounds with respect to loop-freedom. They show that deciding
whether a k-round schedule exists is \textsc{np}-complete, and they present
a polynomial algorithm for computing a weaker variant of loop-freedom.
{\citet{amiri2016transiently} present an \textsc{np}-hardness result for
greedily updating a maximal number of forwarding rules in this context.}
Additionally, \citet{LudwigRFS14} investigate optimal
updates with respect to a stronger property, namely \emph{waypoint
enforcement} in addition to loop freedom.  They produce an update
sequence with a minimal number of waits, using mixed-integer
programming.
\citet{ludwig2016transiently} show that the decision problem is \textsc{np}-hard.

\citet{mattos2016reverse} propose a relaxed variant of per-packet consistency, where a packet may be processed by
several subsequent configurations (rather than a {\em single} configuration), and they present a corresponding
polynomial graph-based algorithm for computing updates.
\citet{dudycz2016can} show that simultaneously computing {\em two} network updates
while requiring a minimal number of switch updates (``touches'') is \textsc{np}-hard.
\citet{brandtconsistent} give a polynomial algorithm to decide if congestion-free
update is possible when flows are ``splittable'' and/or not restricted to
be integer.

\section{Conclusion}
\label{sec:conclusion}

We presented a polynomial-time algorithm to find a consistent update
order on a single packet type.
We then presented a modification to the algorithm, which finds a consistent update order
with a minimal number of waits.
Finally, we proved that this modification is correct, complete, and optimal.
\newpage 

\bibliographystyle{plainnat}
\bibliography{refs}

\end{document}